\documentclass[twocolumn,aps,prr,groupedaddress,superscriptaddress,amsmath,amssymb,showpacs,noeprint]{revtex4-2}

\usepackage{graphicx}
\usepackage{dcolumn}
\usepackage{bm}
\usepackage[dvipsnames]{xcolor}
\usepackage{hyperref}
\hypersetup{colorlinks,linkcolor={MidnightBlue},citecolor={MidnightBlue},urlcolor={MidnightBlue}}  

\usepackage{physics}
\usepackage{amsfonts}
\usepackage{verbatim}
\usepackage{amsmath}
\usepackage{csquotes}
\usepackage{color}
\usepackage{soul}
\usepackage{amsthm}
\usepackage{thmtools}

\newcommand{\mq}[2]{\abs{\ip{#1}{#2}}^2}
\newcommand{\prt}[1]{\left(#1\right)}
\newcommand{\prtq}[1]{\left[#1\right]}
\newcommand{\prtg}[1]{\left\{#1\right\} }
\newcommand{\pdvb}{\partial_\beta}
\newcommand{\pdvbg}{\partial_{\beta_g}}
\newcommand{\pdvbp}{\partial_{\beta_e}}
\newcommand{\pdve}{\partial_\E}

\newcommand{\betag}{\beta_g}
\newcommand{\betae}{\beta_e}

\newcommand{\prtn}[1]{(#1)}
\newcommand{\prtqn}[1]{[#1]}
\newcommand{\prtgn}[1]{\{#1\} }

\newcommand{\ad}{a^\dagger}
\newcommand{\bd}{b^\dagger}

\newcommand{\eH}{\mathcal{H}}
\newcommand{\betagp}{\beta'_g}

\declaretheorem{theorem}
\declaretheorem[sibling=theorem]{lemma}

\newenvironment{equations}
{\begin{equation}\begin{aligned}}
{\end{aligned}\end{equation}\ignorespacesafterend}

\newcommand{\stMin}{\ket{\psi_g}}
\newcommand{\Zmin}{Z_g}
\newcommand{\Emin}{E_g}
\newcommand{\stMax}{\ket{\psi_e}}
\newcommand{\Zmax}{Z_e}
\newcommand{\Emax}{E_e}
\newcommand{\E}{\mathcal{E}}
\newcommand{\Sys}{A\textrm{-}B}
\newcommand{\lambdavec}{{\vec{\lambda}}}
\newcommand{\degen}{d_g}
\newcommand{\degenE}{d_e}
\newcommand{\sig}{\tilde{\sigma}_g}
\newcommand{\lambf}{\vec{\lambda}}

\newcommand{\evn}[2]{\langle #2 | #1| #2 \rangle}
\newcommand{\ketn}[1]{| #1 \rangle}
\newcommand{\dyadn}[1]{| #1 \rangle \langle #1| }

\begin{document}

\title{Energy bounds for entangled states}

\author{Nicol\`o Piccione}
\email{nicolo.piccione@univ-fcomte.fr}
\affiliation{Institut UTINAM, CNRS UMR 6213, Universit\'{e} Bourgogne Franche-Comt\'{e}, Observatoire des Sciences de l'Univers THETA, 41 bis avenue de l'Observatoire, F-25010 Besan\c{c}on, France}

\author{Benedetto Militello}
\affiliation{Universit\`a degli Studi di Palermo, Dipartimento di Fisica e Chimica - Emilio Segr\`{e}, via Archirafi 36, I-90123 Palermo, Italy}
\affiliation{INFN Sezione di Catania, via Santa Sofia 64, I-95123 Catania, Italy}

\author{Anna Napoli}
\affiliation{Universit\`a degli Studi di Palermo, Dipartimento di Fisica e Chimica - Emilio Segr\`{e}, via Archirafi 36, I-90123 Palermo, Italy}
\affiliation{INFN Sezione di Catania, via Santa Sofia 64, I-95123 Catania, Italy}

\author{Bruno Bellomo}
\affiliation{Institut UTINAM, CNRS UMR 6213, Universit\'{e} Bourgogne Franche-Comt\'{e}, Observatoire des Sciences de l'Univers THETA, 41 bis avenue de l'Observatoire, F-25010 Besan\c{c}on, France}

\begin{abstract}

We find the minimum and the maximum value for the local energy of an arbitrary finite bipartite system for any given amount of entanglement, also identifying families of states reaching these bounds and sharing formal analogies with thermal states.
Then, we numerically study the probability of randomly generating pure states close to these energy bounds finding, in all the considered configurations, that it is extremely low except for the two-qubit and highly degenerate cases.
These results can be important in quantum technologies to design energetically more efficient protocols.

\end{abstract}

\maketitle

\textit{Introduction.} Energy and entanglement are two fundamental quantities in physics.
The concept of energy has been of great importance in the
development of physics~\cite{BookFeynman1965} while entanglement is one of the most, if not the most, exotic feature of quantum mechanics~\cite{Einstein1935}.
Therefore, it has been extensively studied since its conception
both from the theoretical and the experimental points of view~\cite{Horodecki2009}, also in connection with nonlocality~\cite{, Bell1964,Freedman1972,Aspect1982,Hensen2015}
and measurements~\cite{BookLandau1981,BookBraginsky1992},
even nondemolitive ones~\cite{Brune1990}.
Entanglement also plays a fundamental role in the development of quantum technologies~\cite{BookNielsen2010} and is considered as a resource in several contexts such as quantum teleportation~\cite{Bennett1993,Bouwmeester1997,Furusawa1998}, quantum cryptography~\cite{Ekert1991,Bennett1992,Gisin2002}, quantum communication~\cite{Duan2001}, quantum computation~\cite{BookNielsen2010}, quantum energy teleportation~\cite{Hotta2014} and in protocols exploiting repeated measurements~\cite{Nakazato2003,Militello2004,Militello2007,Bellomo2009,Bellomo2010}.
As a result, the quest for entanglement generation protocols has been one of the most flourishing fields in recent physics literature~\cite{Yung2005,Fulconis2007,Lin2013,Bellomo2015,Bellomo2017}.

Although quantum algorithms typically make use of two-level systems (qubits) \cite{BookNielsen2010}, it has been shown that $d$-level systems (qudits) may be more powerful for information processing \cite{Bechmann-Pasquinucci2000,Bullock2005}.
Indeed, the higher dimensionality allows for information coding with increased density, leading to a simplification of the design of circuits \cite{Lanyon2009}, since the number of logic gates is reduced. The realization of high-dimensional systems and their control has thus attracted much attention \cite{Neves2005,Kues2017}.

Understanding how energy and entanglement are connected can be crucial in order to design quantum technologies in a more efficient manner~\cite{Chiribella2017}.
In this context, some works investigated the energy cost of generating or extracting entanglement~\cite{Galve2009, Beny2018}.
In particular, some entanglement extraction protocols can be optimized by finding a minimum energy pure state with an assigned entanglement~\cite{Beny2018}. However, this has been done for interacting systems and the explicit solution has been found only for a specific toy model.

In this Rapid Communication, we investigate for an arbitrary finite bipartite system the connection between local energy and entanglement
in the case of discrete local Hamiltonians.
In particular, for any given amount of entanglement, we look for the range of possible values for the local energy and search for quantum states that reach, respectively, the lower and the upper bounds on the local energy.
Moreover, we numerically study the probability of randomly generating pure states close to these energy bounds.
This analysis can be helpful to design energetically efficient entanglement generation protocols. Since the dimensions of the bipartite system are arbitrary, our analysis naturally applies to protocols exploiting qudits.

\textit{Definition of the problem.} We consider a bipartite system $\Sys$ composed of two arbitrary quantum systems $A$ and $B$, with local Hamiltonian $H=H_A+ H_B$, where  $N_A$ and $N_B$ are the dimensions of, respectively,  $H_A$ and $H_B$, being $N_A \leq N_B$.
$H_A$ and $H_B$ can be written as:
\begin{equation}\label{eq:Hamiltonian}
H_X=\sum_{n=0}^{N_X-1} X_n\dyad{X_n},   \quad    X=A ,B ,
\end{equation}
where $ X_0\le X_1\le \dots \le X_{N_X-1}$.

The above local Hamiltonian $H$ suitably describes systems at the start and at the end of most quantum protocols, in which the possible interaction between the subsystems takes place only during the protocol.

We will first consider the case of pure states. In order to quantify the degree of entanglement of a pure  state $\ketn{\psi}$ of  system $\Sys$, we use the entropy of entanglement, which is regarded as the standard entanglement measure for pure states~\cite{Vidal2000,Plenio2007} and is equal to the Von Neumann entropy of one of the reduced states, i.e., $\E (\ketn{\psi})= S(\Tr_{A(B)} \{\dyadn{\psi}\})$,
where $S(\rho)=-\Tr \{\rho \ln \rho\}$.

Every pure state of system $\Sys$ can be rewritten according to its Schmidt decomposition as~\cite{BookNielsen2010}:
\begin{equation}
\label{eq: Schmidt decomposition}
\ket{\psi}=\sum_{i=0}^{N_A-1}\sqrt{\lambda_i}\ket{a_i b_i},
\end{equation}
where
$\ip{a_i}{a_j} =\ip{b_i}{b_j}=\delta_{ij}
 $,  $\sum_{i=0}^{N_A-1} \lambda_i=1$,
and $0\le \sqrt{\lambda_{N_A-1}}\le \dots \le \sqrt{\lambda_1}\le \sqrt{\lambda_0}\le 1$.
Accordingly,  $\E\prt{\ketn{\psi}}= - \sum_i \lambda_i \ln \lambda_i.$

{\it Minimum energy and corresponding states.} For each value of entanglement, $\E$, multiple sets of squared Schmidt coefficients such that the correct amount of entanglement is attained can be found.
Therefore, let us concentrate on one of these sets, $\lambf\equiv\{\lambda_i\}_{i=0}^{N_A-1}$.
In Appendix~\ref{sec: Lowest energy state for a given set of Schmidt coefficients}, we prove Theorem~1, showing that no pure state with the corresponding Schmidt coefficients can have less energy than the state
\begin{equation}
\label{eq: minimal energy state for a given set}
\ket{\psi_{\lambf}} = \sum_{i=0}^{N_A-1} \sqrt{\lambda_i} \ket{A_i B_i},
\end{equation}
having energy
\begin{equation}
E_{\lambf} \equiv \ev{H}{\psi_{\lambf}} = \sum_{i=0}^{N_A-1} \lambda_i E_i, \qquad E_i=A_i+B_i.
\end{equation}
To minimize $E_{\lambf}$ by varying $\lambf$, we use the following bijection (valid up to phase factors on the kets $\ketn{A_i B_i}$):
\begin{equation}
	\label{eq: bijection}
	\ket{\psi_{\lambf}}=\sum_{i=0}^{N_A-1}\sqrt{\lambda_i}\ket{A_i B_i} \leftrightarrow
	\tilde{\rho}_{\lambf}=\sum_{i=0}^{N_A-1} \lambda_i \dyad{A_i B_i},
\end{equation}
from which we get $\E(\ketn{\psi_{\lambf}}) =S(\tilde{\rho}_{\lambf})$.
Moreover, after introducing
\begin{equation}
 \tilde{H} =\sum_{i=0}^{N_A-1} E_i\dyad{A_i B_i},
\end{equation}
we can express the average energy in terms of the density operator $\tilde{\rho}_{\lambf}$ because $\evn{H}{\psi_{\lambf}} =\Tr \{\tilde{H}\tilde{\rho}_{\lambf} \}$.
Thus, the problem of minimizing $E_{\lambf}$ with respect to $\lambf$ for a given degree of entanglement $\E$ is equivalent to finding the diagonal density matrix $\tilde{\rho}_g$ that minimizes energy when its entropy $S=\E$ is fixed.
In Appendix~\ref{sec: Lowest energy state for a given entanglement}, we show that, if  $\E>\ln\degen$ where $\degen\ge 1$ is the number of $\tilde{H}$  eigenstates with lowest energy ($E_{\degen-1}=\dots = E_1=E_0$), the density matrix we search is the thermal state
\begin{equation}
\label{eq: Minimum energy density matrix}
\tilde{\rho}_g = \frac{e^{-\betag \tilde{H}}}{\Zmin}, \qquad \Zmin=\Tr \{e^{-\betag \tilde{H}}\},
\end{equation}
with respect to the fictitious Hamiltonian $\tilde{H}$ and inverse temperature $\betag$ obtained as the positive solution of
\begin{equation}
\label{eq: Thermal solution}
\prt{-\betag \pdv{\betag} + 1}\ln \Zmin = \E.
\end{equation}
In view of Eq.~\eqref{eq: bijection}, this density operator corresponds (up to phase factors) to a minimum energy state given by:
\begin{equation}
\label{eq: minimal energy entangled state}
\stMin=\frac{1}{\sqrt{\Zmin}}\sum_{i=0}^{N_A-1} e^{-\frac{\betag}{2} \prt{A_i + B_i}} \ket{A_i B_i}.
\end{equation}
Its energy can be easily calculated as $\Emin = - \pdvbg \ln{\Zmin}$.
We stress that Eq.~\eqref{eq: Thermal solution} can be solved numerically in a straightforward way and that, in the two-qubit case, analytical expressions can be found.

We observe that the state of Eq.~\eqref{eq: minimal energy entangled state} is not the unique state with minimal energy. Every other state that can be reached from it through the application of local and energy-conserving unitary operators fulfills this request (see Appendix~\ref{sec: Lowest energy state for a given entanglement} for more details).

To conclude our analysis we consider the case $\E \leq \ln \degen$.
In such a situation, the minimum energy is $E_0$ and a minimum energy pure state can be searched in the ground-energy eigenspace so that the problem is trivial.

It is worth stressing that our treatment is valid for every finite $N_A$ and $N_B$, even immensely large.
Therefore, on a physical ground, we conjecture that our analysis holds good even for discrete Hilbert spaces of infinite dimensions, as in the  case of two harmonic oscillators.

\textit{Maximum energy and corresponding states.} The result can be easily obtained by searching for the minimum energy state when considering the Hamiltonians $\bar{H}_{A(B)}=-H_{A(B)}$.
Hence, if $\E > \ln \degenE$, where $\degenE$ is the lowest of the degeneracies of the maximum eigenvalues of $H_A$ and $H_B$,  a maximum energy state is given by
\begin{equation}\label{eq: maximal energy entangled state}
\stMax= \frac{1}{\sqrt{\Zmax}} \sum_{i=0}^{N_A-1} e^{\frac{\betae}{2}\prt{A_i + B_{i+\Delta}}} \ket{A_i B_{i+\Delta}},
\end{equation}
where $\Delta = N_B - N_A$, $\Zmax= \sum_{i=0}^{N_A-1} e^{\betae (A_i + B_{i+\Delta})}$, and $\betae$ is the  positive solution of the equation
$(-\betae\partial_{\betae} + 1)\ln \Zmax = \E$.
Similarly to the minimum energy case, the energy of $\stMax$ can be easily calculated as $\Emax=\pdvbp \ln \Zmax$.

The same considerations made for the minimum energy case about the uniqueness of the state hold good here.
If $\E \leq \ln \degenE$,  then the maximum energy is $A_{N_A-1} + B_{N_B-1}$ and a maximum energy pure state can be searched in the eigenspace of the highest possible energy.

We finally observe that the minimization (maximization) process we have developed
can be easily extended to any other couple of local observables.
Indeed, whatever is the local operator $O=O_A+O_B$ we want to minimize (maximize) for an assigned value of entanglement, we can simply assume that $H_X=O_X$.

\textit{Energy-entanglement distribution.} It is worth commenting at this point about the energy distribution of the states corresponding to the same amount of entanglement.
We have made several numerical simulations finding, in all the studied configurations, that the density of states in the proximity of the bounding curves is
extremely low, except for the two-qubit case and highly degenerate cases.
In fact, the main part of the states occupy the intermediate region, and the discrepancy between the peripheral and central densities becomes higher and higher as the dimensionality of the systems increases.
We report here, as an example, the density of states corresponding to two local Hamiltonians having spectra given by $\sigma(H_A)=\prtg{0,2,4}$ and $\sigma(H_B)=\prtg{0,1,6,9}$ in arbitrary units.
In particular, in Fig.~\ref{fig: comparazione1} we show the two curves defining the energy bounds for assigned entanglement and the distribution of a large number of randomly generated pure states~\cite{,Miszczak2011} (the behavior of $\betag$ and $\betae$ is shown in Appendix~\ref{sec: Lowest energy state for a given entanglement}).
It is well visible that the majority of the states lies in the central zone, while none of the generated states is very close to the bounding curves.
This circumstance allows one to better appreciate the relevance of our results since, for example, in an entanglement generation process, one could choose to generate the state $\ketn{\psi_g}$ having the lowest energy for the desired amount of entanglement, instead of any of all the other states which require more energy.
We finally observe that the randomly generated states numerically satisfy the known theoretical expected averages both in entanglement and energy~\cite{Page1993,Foong1994,BookGemmer2009}.

\begin{figure}
	\centering
	\includegraphics[width=0.48\textwidth]{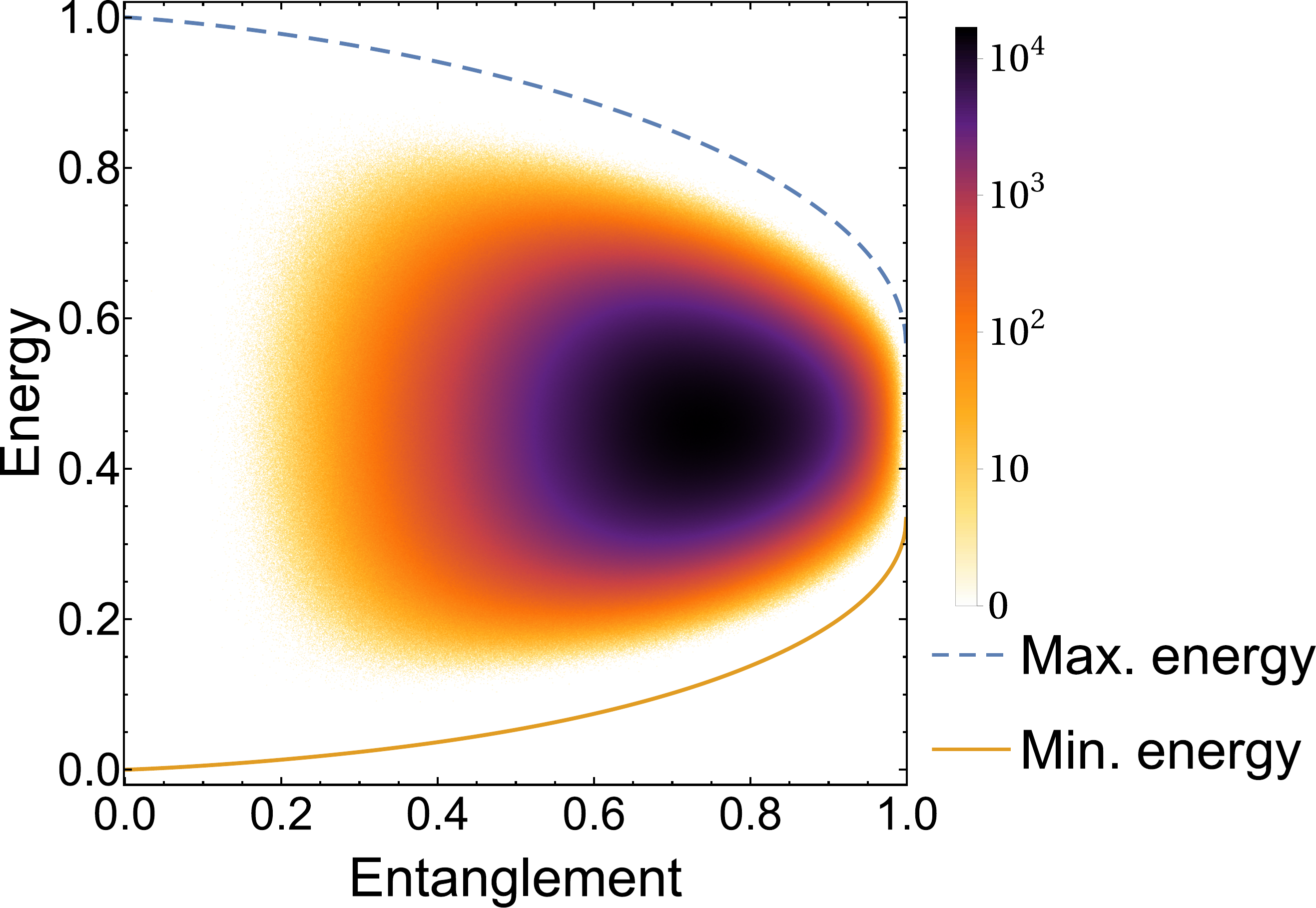}
	\caption{
		Distribution of $10^9$ randomly generated pure states with respect to the entropy of entanglement and the local energy in a $1000\times 1000$ grid. The  Hamiltonians have spectra: $\sigma(H_A)=\prtg{0,2,4}$ and $\sigma(H_B)=\prtg{0,1,6,9}$ in arbitrary units.
		Both the entanglement and the energy are normalized with respect to their maxima.
	}
	\label{fig: comparazione1}
\end{figure}

\textit{Two-qubit system.} Now we apply our general results to the case of two qubits, i.e., to the case $N_A=N_B=2$. By using the purity $P$ [where $P(\rho)= \Tr \{\rho^2\}$] of one of the reduced states  instead of the entropy of entanglement $\E$ as entanglement quantifier, it is possible to obtain through straightforward calculations closed analytical expressions both for the minimum and maximum energy states and for the energy bounds using Eqs.~\eqref{eq: minimal energy entangled state} and \eqref{eq: maximal energy entangled state}. This is possible thanks to the fact that for a two-qubit system the Von Neumann entropy and the purity can be bijectively connected.
Starting from $\stMin=\sqrt{\lambda}\ketn{A_0 B_0}+\sqrt{1-\lambda}\ketn{A_1 B_1}$ and imposing $(1-\lambda)/\lambda = \exp[-\betag(E_1-E_0)]$,
one can easily obtain $\betag = - \prtn{E_1-E_0}^{-1}  \ln \prtqn{(1-\lambda)/\lambda}$, where $\lambda= \prtn{1+\sqrt{2 P-1}}/2$. Analogously, one can find $\betae= \betag$.
Moreover, we can express the energy bounds as $\Emin = \lambda E_0 + (1-\lambda)E_1$ and $\Emax =(1-\lambda) E_0 + \lambda E_1$.

\textit{Mixed states.} We now show that the bounds derived above are still valid even
when we extend the analysis to mixed states.
Contrarily to the pure state case, a standard entanglement quantifier does not exist~\cite{Horodecki2009}.
However, it is in general required that the convexity property is satisfied~\cite{Vidal2000,Plenio2007}, i.e.,
for any arbitrary quantifier $\E_m$
\begin{equation}
\label{eq: stronger LOCC assumption}
\rho = \sum_i p_i \rho_i \implies \E_m (\rho) \leq \sum_i p_i \E_m (\rho_i),
\end{equation}
where $p_i\geq 0\ \forall i$ and $ \sum_i p_i = 1$.
In addition, we make the standard assumption that $\E_m$ applied to pure states is equal to the entropy of entanglement~\cite{Plenio2007}.
In Appendix~\ref{sec: The convexity problem}, we show that this assumption can be relaxed.

Every mixed state can be written as a combination of pure states, $\rho=\sum_i p_i \dyadn{\psi_i}$. Thus, every mixed state has energy equal to $\Tr \prtg{H\rho}=\sum_i p_i \evn{H}{\psi_i}$
and entanglement $\E_m (\rho) \leq \sum_i p_i \E_i$, where $\E_i =\E_m(\ketn{\psi_i})$.
Since one can prove (see Appendix~\ref{sec: The convexity problem}) that the curves $\Emin(\E)$ and $\Emax(\E)$ are, respectively, the former increasing and convex, and the latter decreasing and concave, the following chain of relations holds:
\begin{equation}
\Tr \prtg{H\rho}  \ge \sum_i p_i \Emin(\E_i)
\ge \Emin\prt{\sum_i p_i \E_i} \ge \Emin (\E_m (\rho)).
\end{equation}
Analogously, it holds that $\Tr \prtg{H\rho}\le \Emax (\E_m (\rho))$.

It follows that, in an energy-entanglement graph, every mixed state can be found on a segment that is entirely between the minimum and maximum energy curves. In Fig.~\ref{fig: ConvexityGraph} an example of this situation is clearly shown.

\begin{figure}
	\centering
	\includegraphics[width=0.48\textwidth]{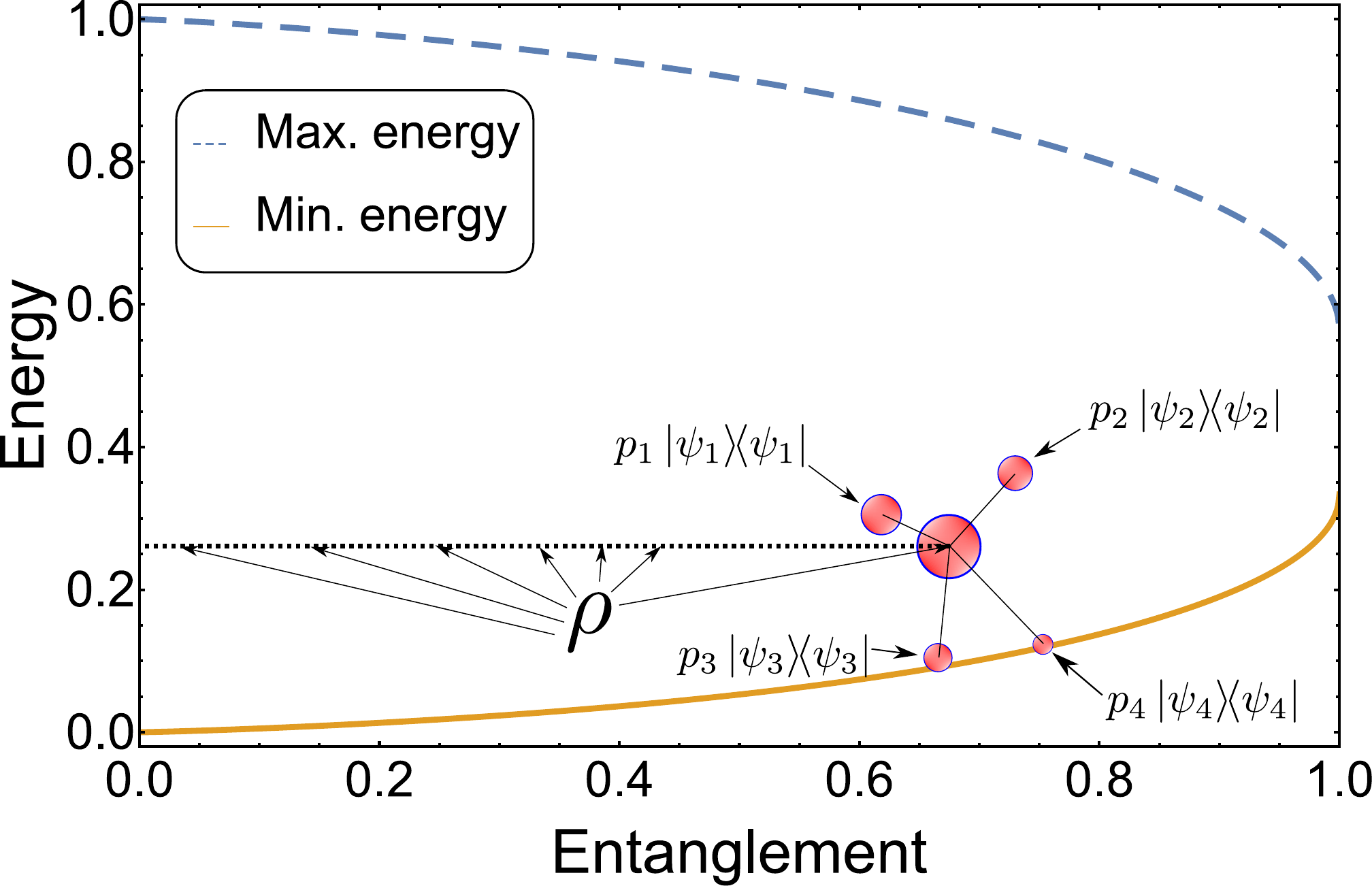}
	\caption{
		Representation of the energy-entanglement relation for a mixed state obtained as a convex sum of four pure states. The assigned energy value (obtained as the average of the energies of the pure states) and the possible values of entanglement (from zero to the average of the entanglement values of the single pure states) identify a segment.
		This segment always lies between the curves of minimum and maximum energy because of their monotonicity and convexity properties.
	}
	\label{fig: ConvexityGraph}
\end{figure}

{\it Connections with thermodynamics, entanglement Hamiltonian, and LOCC.}
The minimum and maximum energy states are characterized by coefficients that can be directly linked to the Boltzmann factors of a fictitious thermal state and, as a consequence, their energy can be calculated through their fictitious partition function.
This is worth mentioning because entanglement and thermodynamics are believed to be conceptually connected in the context of typicality~\cite{Popescu2006,BookGemmer2009} and they have various formal analogies when treated within resource theories such as local operations and classical communication (LOCC) and thermodynamic resource theory (TRT)~\cite{Brandao2008,Horodecki2009,Horodecki2013}.
In fact, a connection with thermodynamics has been also found in the study conducted in Ref.~\cite{Bakhshinezhad2019}. There,  the authors dealt with the problem of creating the maximum amount of correlations (quantified by mutual information) by employing a limited amount of energy, through the application of a unitary operator. They considered non-interacting bipartite systems starting from thermal product states.
In the zero temperature limit, since the mutual information is twice the entropy of entanglement, their problem coincides with our search for the minimum energy states for a given amount of entanglement. Indeed, they have found that to maximize the correlations one has to generate states of the form of Eq.~\eqref{eq: minimal energy entangled state}~\footnote{We note that Ref.~\cite{Bakhshinezhad2019} appeared on the arXiv during the review process}.
Their proof relies on the concept of passive states (states the energy of which cannot be lowered by unitary operations), thus providing an additional link between our results and the field of thermodynamics. 

It is also interesting to consider the limit case of Eq.~\eqref{eq: minimal energy entangled state} when $B_{N_A-1}= B_{N_A-2}= \dots = B_0$. In this case, the reduced state of $A$,  $\rho_g^A\equiv  \Tr_B \{\dyadn{\psi_g}\}$, is equal to
\begin{equation}
\rho_g^A = \frac{1}{\Zmin}\sum_{i=0}^{N_A-1} e^{-\betag A_i} \dyad{A_i},
\quad
\Zmin = \sum_{i=0}^{N_A-1} e^{-\betag A_i},
\end{equation}
which is a thermal state with respect to $H_A$ at temperature $T=1/(k_B \betag)$, where $k_B$ is the Boltzmann constant.
This result can be easily obtained without using  Eq.~\eqref{eq: minimal energy entangled state} since in this limit the problem reduces to find the minimum energy state for a fixed entropy of subsystem $A$.

Our results also present connections with some studies based on the entanglement Hamiltonian formalism.
Indeed, when $N_A=N_B$, the reduced states of $\stMin$ are ($\rho_g^B\equiv \Tr_A \{ \dyadn{\psi_g}\}$)
\begin{equation}
\rho_g^{A(B)} = \frac{1}{\Zmin} e^{-\betag \Tr_{B(A)} \{\tilde{H}\}}.
\end{equation}
Thus, the reduced states have been easily written in the entanglement Hamiltonian formalism (this can be done for  $\rho_g^A$ even when $N_A < N_B$), which has been proved to be useful to get various insights into solid-state physics research \cite{Li2008,Peschel2009,Peschel2011,Laflorencie2016,Dalmonte2018}.
In Appendix~\ref{sec: Many-body systems with minimal energy states as ground states}, we show that some many-body systems of interest are characterized, under appropriate approximations~\cite{Peschel2011}, by ground states belonging to the family of minimal energy states for a given entanglement.

Lastly, we point out that minimum or maximum energy states with respect to the same local Hamiltonians are connected through one-shot LOCC (see Appendix~\ref{sec: Minimum energy states are LOCC-connected} for the proof).
This  has two direct implications. The first one is that, given any pure state $\ketn{\psi}$, one can immediately write down a family of states that are LOCC-connected to it  (see Appendix~\ref{sec: Minimum energy states are LOCC-connected}). 
The second one is that if two distant parties share a minimum energy state having more entanglement than needed,  they can recover, with certainty, the maximum amount of local energy compatible with the needed entanglement.

\textit{Connections to quantum technologies.} Our results are particularly relevant in protocols exploiting partially entangled qudits. Although maximally entangled states are requested in many applications, non-maximally entangled states have been proven useful in quantum technologies, for example in processes involving two-mode squeezed states~\cite{Takei2005,Adesso2014}, in quantum telecloning of qudits~\cite{Gordon2007,Araneda2016}, and in probabilistic quantum teleportation~\cite{Banaszek2000}.
In the last two cases, our results allow one to implement the procedure by exploiting less expensive entangled states, through the direct utilization of minimum energy states or of  Theorem~1 (see Appendix~\ref{sec:Theorem1Application}).

More in general, in the LOCC asymptotic limit, $n$ copies of a state $\ketn{\phi}$ can be converted to $m$ copies of a state $\ketn{\phi'}$ if and only if $n \E (\ketn{\phi}) \geq m \E (\ketn{\phi'})$, with $n,m \rightarrow \infty$, thus making the entropy of entanglement the quantifier of the resource entanglement~\cite{BookNielsen2010}. For example, Bell states can always be obtained by entanglement distillation~\cite{BookNielsen2010}.
In this framework, given a certain amount of energy, it is then particularly relevant that
it is possible to generate more entanglement overall by producing many copies of our minimum energy states with non maximal entanglement (see Appendix~\ref{sec:EntanglementProduction}).

Our results also permit one to identify bounds in the production of pure entangled states within the framework of the TRT, which has recently drawn a lot of attention~\cite{Horodecki2013,Lostaglio_Tesi}.
Its goal is to study what states are reachable through thermal operations given an arbitrary starting state $\rho$ and the environmental temperature $T$.
Since the energy amount of reachable states from the state $\rho$ is bounded, when TRT is equipped with our results, it lets us individuate which are the reachable pure states with the maximum allowed degree of entanglement.
Indeed, allowing the use of catalysts~\cite{Lostaglio_Tesi}, the state we search is one of our minimum energy states with energy equal to $\Tr \{\rho (H_A+H_B)\} - k_B T S (\rho)$.

We have also proven, under the conjecture that our analysis is valid also in the case of discrete Hilbert spaces of infinite dimensions, that two-mode squeezed states are minimum energy states for a given amount of entanglement (see Appendix~\ref{sec: The minimum energy state of two harmonic oscillators is a two-mode squeezed state}).
Therefore, these states, extensively exploited in quantum optics laboratories~\cite{Schumaker1985,Takei2005,Adesso2014,Dutt2015}, are the most energetically convenient states to generate.
In general, a possible way to generate minimum energy states is to exploit dissipative processes leading to a unique steady state, such as simple zero-temperature thermalizations~\cite{BookBreuer2002}.
In this case, choosing a suitable interaction Hamiltonian leads the bipartite system to the desired state, i.e., the ground state.
We give an example of this process for a two-qubit system and for a two-harmonic-oscillator system in Appendix~\ref{sec: Generation of a minimum energy state through zero temperature thermalization}.
Such kind of processes involving a simple thermalization have been studied, for example, in Ref.~\cite{Piccione2019Work}.

\textit{Conclusive remarks.} In summary, we have found the minimum and maximum permitted local energy of an arbitrary finite bipartite system for a given quantity of entanglement,
also reporting the explicit form of a family of minimum and maximum energy states.
Then, we have numerically investigated the energy distribution of entangled pure states, finding, in all the studied configurations, that the probability of randomly generating states with a fixed entanglement close to the energy bounds is extremely low except for the two-qubit case and highly degenerate cases.

Our results can be important in quantum technologies since, given the degree of entanglement necessary for a certain application, our approach allows one to identify a class of states whose generation requires the lowest energy cost. Such an identification appears even more important also in the light of our numerical simulations, showing that the energies of the majority of the states with a fixed entanglement typically lie quite far from the energy bounds.
Finally, we stress that Theorem~1 can bring by itself great practical advantages in optimization problems depending exclusively on the Schmidt coefficients, given some energy  constraints, as discussed in detail in Appendix~\ref{sec:Theorem1Application}.

\textit{Acknowledgements.} N.P. thanks Mauro Paternostro, Andrea Smirne, Jan Sperling,  and Alexander Streltsov for useful discussions about the results of this Rapid Communication.

\appendix

\section{\label{sec: Lowest energy state for a given set of Schmidt coefficients}  Lowest energy state for a given set of Schmidt coefficients}

In this section, we prove the following theorem.

\begin{theorem}
	\label{theorem: minimal energy state for a given set of Schmidt coefficients}
	For any bipartite system with Hamiltonians $H_A$ and $H_B$ of the form of Eq.~$(1)$ of the main text,  given a fixed set of squared Schmidt coefficients $\lambdavec\equiv\prtg{\lambda_i}_{i=0}^{N_A-1}$ with $\lambda_i\leq\lambda_j$ for $ i>j$, no pure state can have less energy than the state
	\begin{equation}
	\label{eq: minimal energy Schmidt state}
	\ket{\psi_{\lambf}}=\sum_{i=0}^{N_A-1} \sqrt{\lambda_i} \ket{A_i B_i}.
	\end{equation}
	Moreover, if both $H_A$ and $H_B$ do not have degeneracies and $\lambda_i < \lambda_j$ for $ i>j$, the above state  is the only pure state with that energy up to phase factors on the basis kets.
\end{theorem}

The idea is to prove that, being $\ketn{\psi} = \sum_{i=0}^{N_A-1} \sqrt{\lambda_i} \ketn{a_i b_i}$ any other pure state with the same coefficients, it holds
\begin{equation}
\ev{H}{\psi} - \ev{H}{\psi_{\lambf}}= \sum_{i=0}^{N_A-1} \lambda_i \Delta_i \geq 0,
\end{equation}
where $\Delta_i=\evn{H}{a_i b_i} - \evn{H}{A_i B_i}$.
In order to prove our statements, we will need four lemmas.

\begin{lemma}
	\label{lemma: Minim. Somma Decrec. Cresc.}
	Consider two sets of real quantities $\prtg{p_i}_{0}^{N-1}$ and $\prtg{E_i}_{0}^{N-1}$, where $0\leq p_i \leq A \ \forall i$ and $\sum_{i=0}^{N-1} p_i = MA$, with $M \leq N,\ A\in \mathbb{R}^+$.
	Then,
	\begin{equation}\label{eq: lemma2a}
	\sum_{i=0}^{N-1} p_i E_i
	\geq
	\sum_{i=0}^{N-1} p_i^{\downarrow} E_i^{\uparrow}
	\geq
	A \sum_{i=0}^{M-1} E_i^{\uparrow},
	\end{equation}
	where $\prtg{p_i^{\downarrow}}_{0}^{N-1}$ is the set $\prtg{p_i}_{0}^{N-1}$ with the elements put in decreasing order and $\prtg{E_i^{\uparrow}}_{0}^{N-1}$ is the set $\prtg{E_i}_{0}^{N-1}$ with the elements put in increasing order.
	Moreover, if the set $\prtg{E_i}_{0}^{N-1}$ has no repeated values and at least one $p_i>0$ with $i\geq M$, then
	\begin{equation}
	\sum_{i=0}^{N-1} p_i E_i > A \sum_{i=0}^{M-1} E_i^{\uparrow}.
	\end{equation}
\end{lemma}

\begin{proof}
	In the first part of the proof we show that
	$\sum_{i=0}^{N-1} p_i E_i \geq \sum_{i=0}^{N-1} p_i^{\downarrow} E_i^{\uparrow}$.
	Without loss of generality, we can first put the $p_i$ in decreasing order and continue to use $E_i$ to indicated the elements of the permutated set.
	Then, if $E_n < E_m$, with $m<n$ we have
	\begin{equation}
	\sum_{i=0}^{N-1} p_i^{\downarrow}  E_i
	\geq
	\sum_{i\neq n,m}^{N-1} p_i^{\downarrow}  E_i + p_m^{\downarrow}  E_n + p_n^{\downarrow}  E_m.
	\end{equation}
	The possibility of iterating this procedure concludes the first part of the proof.
	
	For the second part of the proof, we consider the two sets already in the correct order  (decreasing for the $p_i$ and increasing for the $E_i$) and we avoid the arrows to lighten the notation. We consider the following iterative procedure.
	If $p_0 + p_{N-1} < A$, we write
	\begin{equation}
	\label{eq: lowering passage}
	\sum_{i=0}^{N-1} p_i E_i
	\geq
	\sum_{i=1}^{N-2} p_i E_i + (p_0+p_{N-1}) E_0=
	\sum_{i=0}^{N-2} q_i E_i,
	\end{equation}
	and repeat the procedure with the new set $\prtg{q_i}_{0}^{N-2}$, where $q_i=p_i\ \forall\ 1\leq i \geq N-2$ and $q_0 = p_0 + p_{N-1}$.
	Also, notice that the above inequality becomes strict if $E_{N-1}> E_0$ and $p_{N-1}>0$.
	Otherwise, we have $p_0 + p_{N-1} = A + \tilde{p}_{N-1} \ge A$ and
	\begin{equations}
		\label{eq: term A passage}
		&\sum_{i=0}^{N-1} p_i E_i\\
		&=\sum_{i=1}^{N-2} p_i E_i + p_0 E_0 + \prt{A-p_0}E_{N-1}+\tilde{p}_{N-1}E_{N-1} \\
		&\geq
		\sum_{i=1}^{N-2} p_i E_i + A E_0 + \tilde{p}_{N-1} E_{N-1} =
		A E_0 + \sum_{i=1}^{N-1} \tilde{p}_i E_i,
	\end{equations}
	where, in the new set $\prtg{\tilde{p}_i}_{1}^{N-1}$, $\tilde{p}_i=p_i\ \forall\ 1\leq i \leq N-2$ and we recall that $\tilde{p}_{N-1} = p_0 + p_{N-1} - A$.
	In this case, the next step has to be done on $\sum_{i=1}^{N-1} \tilde{p}_i E_i$.
	
	The step represented by Eq.~\eqref{eq: lowering passage} conserves the sum of the $p_i$, i.e., $\sum_{i=0}^{N-1} p_i = \sum_{i=0}^{N-2} q_i$, while in Eq.~\eqref{eq: term A passage}, $\sum_{i=0}^{N-1} p_i = \sum_{i=1}^{N-1} \tilde{p}_i + A$.
	At each step of the procedure, we lose an element of the sum. Thus, at the end of the procedure, only the terms $A E_i$ produced by the steps such as the one of Eq.~\eqref{eq: term A passage} survive.
	The number of times this kind of step takes place is equal to $M$ because of the hypothesis $\sum_{i=0}^{N-1} p_i = MA$ and, in the end, we will get $A \sum_{i=0}^{M-1} E_i$, which concludes the proof of the validity of Eq.~\eqref{eq: lemma2a}.
	The particular case of no degeneracies in the set $\prtg{E_i}_{0}^{N-1}$ and at least one $p_i>0$ with $i\geq M$ follows by considering that, in this case, at least one passage of Eq.~\eqref{eq: lowering passage} or of Eq.~\eqref{eq: term A passage} with the strict inequality has to be performed.
\end{proof}

\begin{lemma}
	\label{lemma: Projection Lemma}
	Consider an Hamiltonian of the form $H=\sum_{i=0}^{N-1} E_i \dyadn{E_i}$, where $E_i \leq E_j$ for $ i < j$ and a set of orthonormal vectors on the same Hilbert space $\prtgn{\ketn{a_i}}_{i=0}^{M-1}$.
	It holds that
	\begin{equation}
	\sum_{i=0}^{M-1} \ev{H}{a_i} \geq \sum_{i=0}^{M-1} E_i \quad \forall M \leq N.
	\end{equation}
	Moreover, if the spectrum of the Hamiltonian is non degenerate, the equality sign is obtained if and only if we can write $\ketn{a_i} = \sum_{n=0}^{M-1} \alpha_{i, n}\ketn{E_n}\ \forall i$.
\end{lemma}

\begin{proof}
	Let us start by considering that
	\begin{equation}
	\sum_{i=0}^{M-1} \ev{H}{a_i} =
	\sum_{n=0}^{N-1} p_n E_n,
	\end{equation}
	where $p_n = \sum_{i=0}^{M-1} \mq{E_n}{a_i} \leq 1$.
	Moreover, $\sum_{n=0}^{N-1} p_n = M$.
	Then, because of lemma \ref{lemma: Minim. Somma Decrec. Cresc.}
	\begin{equation}
	\label{eq: strict lemma 3}
	\sum_{i=0}^{N-1} p_n E_n \geq \sum_{n=0}^{M-1} E_n,
	\end{equation}
	which concludes the first part of the proof.
	The second part follows by considering that if a ket $\ketn{a_i}$ exists such that it cannot be obtained as a linear combination of the first $M$ energy eigenvectors, then at least one $p_n>0$ with $n\geq M$ exists. Then, because of lemma \ref{lemma: Minim. Somma Decrec. Cresc.} the above inequality is strict.
	On the other hand, if the conditions on the kets $\ketn{a_i}$ are valid, the equality sign of Eq.~\eqref{eq: strict lemma 3} is trivially obtained.
\end{proof}

Let us now consider a set of $D$ real numbers  $\prtg{\Delta_i}_{0}^{D-1}$  such that $\sum_{i=0}^{N-1} \Delta_i \geq 0, \quad \forall N\leq D$. The following lemmas (\ref{lemma: Series composition} and \ref{lemma: Strict series composition}) are valid.

\begin{lemma}
	\label{lemma: Series composition}
	Given a set of real non-negative numbers $\lambda_i $ such that  $\lambda_i\leq\lambda_j$ for $i>j$ then
	\begin{equation}
	\sum_{i=0}^{N-1} \lambda_i \Delta_i \geq 0, \quad \forall N\leq D.
	\end{equation}
\end{lemma}

\begin{proof}
	We show this by induction.
	Obviously, $\lambda_0 \Delta_0 \geq 0$.
	Suppose that the lemma is true for $M<D$, that is
	\begin{equation}
	\sum_{i=0}^{N-1} \lambda_i \Delta_i \geq 0, \quad \forall N\leq M < D.
	\end{equation}
	We have to analyze
	\begin{equation}
	\label{eq: lemma series composition condition}
	\lambda_0 \Delta_0 + \dots \lambda_{M-1} \Delta_{M-1} + \lambda_{M} \Delta_{M} \geq 0.
	\end{equation}
	If $\Delta_M$ is non-negative, the result is trivial.
	If $\Delta_M$ is negative we have
	\begin{multline}
	\label{eq: lemma 4 dis1}
	\lambda_0 \Delta_0 + \dots + \lambda_{M-1} \Delta_{M-1} + \lambda_{M} \Delta_{M}
	\geq\\
	\geq
	\underbrace{\lambda_0 \Delta_0 + \dots}_{\textrm{True}} + \lambda_{M-1} (\Delta_{M-1} + \Delta_{M}).
	\end{multline}
	If $(\Delta_{M-1} + \Delta_{M})\geq 0$, Eq.~\eqref{eq: lemma series composition condition} is satisfied.
	Otherwise, we go on and consider
	\begin{equation}
	\label{eq: lemma 4 dis2}
	\underbrace{\lambda_0 \Delta_0 + \dots}_{\textrm{True}} + \lambda_{M-2} (\Delta_{M-2} + \Delta_{M-1} + \Delta_{M}) \geq 0.
	\end{equation}
	In the worst case we arrive at
	\begin{equation}
	\lambda_0 \prt{\sum_{i=0}^{M} \Delta_i} \geq 0,
	\end{equation}
	which is true by hypothesis.
	This concludes the proof.
\end{proof}

\begin{lemma}
	\label{lemma: Strict series composition}
	Given a set of real positive numbers $\lambda_i $ such that  $\lambda_i < \lambda_j$ for $i>j$, if there exists $ \ i_0 = \min_{i}\prtg{ i : \Delta_{i}>0}$ then
	\begin{equation}
	\sum_{i=0}^{N-1} \lambda_i \Delta_i > 0, \quad \forall \ \prt{i_0 +1} \leq N \leq D.
	\end{equation}
\end{lemma}

\begin{proof}
	Of course, $\lambda_{i_0}\Delta_{i_0} > 0$.
	Then, we can repeat the reasoning of the proof of lemma \ref{lemma: Series composition}, keeping into account that when one inequality such as that of Eq.~\eqref{eq: lemma 4 dis1}   has to be considered, it will be a strict inequality.
\end{proof}

Now we are ready to prove theorem~\ref{theorem: minimal energy state for a given set of Schmidt coefficients}.
\begin{proof}
	Let be $\ketn{\psi} = \sum_{i=0}^{N_A-1} \sqrt{\lambda_i} \ketn{a_i b_i}$ an arbitrary pure state with the given set of Schmidt coefficients.
	Let us calculate:
	\begin{equation}
	\ev{H}{\psi} - \ev{H}{\psi_{\lambf}}= \sum_{i=0}^{N_A-1} \lambda_i \Delta_i,
	\end{equation}
	where
	\begin{equations}
		\Delta_i
		&=\ev{H}{a_i b_i} - \ev{H}{A_i B_i}\\
		&=\ev{H_A}{a_i}-A_i +\ev{H_B}{b_i} - B_i.
	\end{equations}
	Because of lemma \ref{lemma: Projection Lemma},
	\begin{equation}
	\sum_{i=0}^{N-1} \Delta_i \geq 0, \ \ \forall \ N \leq N_A .
	\end{equation}
	Then, because of lemma \ref{lemma: Series composition}, $\sum_{i=0}^{N_A-1} \lambda_i \Delta_i \geq 0$ and the first part of the theorem is proven.
	
	For the second part of the theorem, because of lemma~\ref{lemma: Strict series composition}, the energy of the arbitrary state $\ketn{\psi}$ is equal to the energy of $\ketn{\psi_{\lambf}}$ if and only if $\Delta_i = 0$ for each $i$.
	Starting from $\Delta_0$, the only way to make it zero is to set $\ketn{a(b)_0}=\ketn{A(B)_0}$, up to phase factors, because $A(B)_0$ is the lowest eigenvalue.
	Then, the only way to set $\Delta_1 = 0$ is to set $\ketn{a(b)_1}=\ketn{A(B)_1}$, up to phase factors, because this ket has to be orthogonal to $\ketn{a(b)_0}$.
	The continuation of this reasoning leads to the conclusion of the proof.
\end{proof}

\section{\label{sec: Lowest energy state for a given entanglement}    Lowest energy state for a given entanglement}

Here, we prove the main result of the main text, i.e., to find one pure state of minimum energy for a fixed amount of entanglement, quantified by the entropy of entanglement $\E (\ketn{\psi})$.

Because of the results obtained in the previous section, we know that one minimum energy pure state can be searched among pure states of the form of Eq.~\eqref{eq: minimal energy Schmidt state}.
A bijection between these states and a specific set of diagonal density matrices exists up to phase factors on the kets $\ketn{A_i B_i}$:
\begin{equation}
\label{eq: bijection APP}
\ket{\psi_{\lambf}} \longleftrightarrow
\tilde{\rho}_{\lambf} = \sum_{i=0}^{N_A-1} \lambda_i \dyad{A_i B_i}.
\end{equation}
It also holds:
\begin{equations}
	\ev{(H_A + H_B)}{\psi_{\lambf}} &= \Tr \prtg{\tilde{H}\tilde{\rho}_{\lambf}}, 
	\\
	\E (\ket{\psi_{\lambf}})= S(\tilde{\rho}_{\lambf}) & = -\sum_{i=0}^{N_A-1} \lambda_i \ln \lambda_i,
\end{equations}
where $\tilde{H}=\sum_{i=0}^{N_A-1} E_i\dyadn{A_i B_i}$, using the shorthand notation $E_i=A_i+B_i$, and  $S(\rho)=-\Tr \prtg{\rho \ln \rho}$.
Thus, the problem of minimizing the energy of $\ketn{\psi_{\lambf}}$ with respect to the sets $\lambf=\prtg{\lambda_i}_{i=0}^{N_A-1}$ such that $\E (\ketn{\psi_{\lambf}})= \E$ is equivalent to find the diagonal density matrix $\tilde{\rho}_{\lambf}$ that minimizes energy when its entropy $S(\tilde{\rho}_{\lambf})=\E$ is fixed.

Let us suppose that the ground energy level of $\tilde{H}$ is degenerate with degeneration $\degen$.
Then, trivially, when the entanglement is lower or equal to $\ln \degen$, the minimum energy is $E_0$ and a minimum energy  state has to be found inside the degenerate subspace of energy $E_0$.
In the other case, i.e., $\E > \ln \degen$, we show in the following that the solution is given by
\begin{equation}
\label{eq: Minimum energy density matrix APP}
\tilde{\rho}_{g} = \frac{e^{-\betag \tilde{H}}}{\Zmin}, \qq{where} \Zmin=\Tr \prtg{e^{-\betag \tilde{H}}},
\end{equation}
and $\betag$ is the non-negative solution of the equation:
\begin{equation}
\label{eq: Thermal solution APP}
\prt{-\betag \pdv{\betag} + 1}\ln \Zmin = \E.
\end{equation}

Suppose that it exists a state $\sig$ with the same entropy of $\tilde{\rho}_g$ but lower energy.
To each state $\rho$ we can associate a functional formally equivalent to the free energy:
\begin{equation}
F(\rho,\tilde{H},\betag) = \Tr \prtg{\tilde{H}\rho} - \frac{1}{\betag}S(\rho).
\end{equation}
Then, it holds
\begin{equation}
F(\sig,\tilde{H},\betag) - F(\tilde{\rho}_g,\tilde{H},\betag) =  \Tr \prtg{\tilde{H}\sig} - \Tr \prtg{\tilde{H}\tilde{\rho}_g} <0,
\end{equation}
but also, as $\tilde{\rho}_g$ is the thermal state with inverse temperature $\betag$,
\begin{equation}
F(\sig,\tilde{H},\betag) - F(\tilde{\rho}_g, \tilde{H},\betag) = \frac{1}{\betag} S(\sig||\tilde{\rho}_g) > 0,
\end{equation}
where $S(\rho ||\sigma)= \Tr \prtg{\rho\prt{\ln \rho - \ln \sigma}}$, which is always positive for $\sigma \neq \rho$~\cite{BookBreuer2002}.
Therefore, $\sig$ cannot exist.

It is straightforward to prove that Eq.~\eqref{eq: Thermal solution APP} always has a unique non-negative solution, because the entropy of entanglement  is a continuous and strictly decreasing function of $\betag$  [see Eq.~\eqref{eq:derivEnt} of the following section], and, moreover, the thermal state assumes the minimum and maximum values of entropy in the two limit cases $\betag=0$ ($\E =\ln N_A$) and $\betag\rightarrow\infty$ ($\E\rightarrow \ln \degen$).

On the basis of Eqs.~\eqref{eq: bijection APP} and \eqref{eq: Minimum energy density matrix APP}, we finally get:
\begin{equation}
\label{eq: minimal energy entangled state APP}
\stMin=\frac{1}{\sqrt{\Zmin}}\sum_{i=0}^{N_A-1} e^{-\frac{\betag}{2}E_i} \ket{A_i B_i},
\end{equation}
as one possible minimum energy state. Note that a whole family of pure states with this amount of local energy and this entropy of entanglement can be easily constructed from Eq.~\eqref{eq: minimal energy entangled state APP} by multiplying kets by single phase factors:
\begin{equation}
\frac{1}{\sqrt{\Zmin}}\sum_{i=0}^{N_A-1} e^{-\frac{\betag}{2}E_i} e^{i \alpha_i} \ket{A_i B_i}, \qq{with} \alpha_i \in [0,2\pi[.
\end{equation}
Because of the  theorem~\ref{theorem: minimal energy state for a given set of Schmidt coefficients}, we can say that this family comprehends all the lowest energy pure states if the two local Hamiltonians $H_A$ and $H_B$ have no degeneracies (we recall  that $\lambda_i= \exp\prtg{-\betag E_i/2}$, therefore the condition $\lambda_i > \lambda_j$ for $i<j$ is satisfied).

If at least one of the Hamiltonians has degeneracies, other states are valid.
For example, suppose that $A_m = A_{m+1}$.
Then, the state obtained from Eq.~\eqref{eq: minimal energy entangled state APP} by swapping $\ketn{A_m}$ with $\ketn{A_{m+1}}$ is still a lowest energy state.
In general, every state that can be obtained from the state of Eq.~\eqref{eq: minimal energy entangled state APP} through the application of energy preserving unitary operators of the form $U_A \otimes U_B$ is one of the lowest energy states.

The search for the maximum energy bound and the corresponding maximum energy state can be easily performed analogously to the minimum energy case, by considering the Hamiltonians $\bar{H}_{A(B)}=-H_{A(B)}$. The results are reported in the main text in terms of the parameter $\betae$. In general $\betag$ and $\betae$ are different. However, it is possible to show that $\betag=\betae$ when the local spectra eigenvalues are symmetric with respect to a rotation (i.e., a real constant $C$ exists such that the spectra $\sigma \{H_{A(B)}\} = -\sigma \{H_{A(B)}\} + C $).
This symmetry is automatically satisfied in the case of two qubits examined in the main text.

We finally observe that, in general, the solution of Eq.~\eqref{eq: Thermal solution APP} can be easily computed numerically.
In Fig.~\ref{fig: grafBeta1}, we plot $\betag$ and $\betae$ as functions of $\E$, for the specific system (a simple $3\times4$ system) analyzed in the main text in the part energy-entanglement distribution.

\begin{figure}
	\centering
	\includegraphics[width=0.48\textwidth]{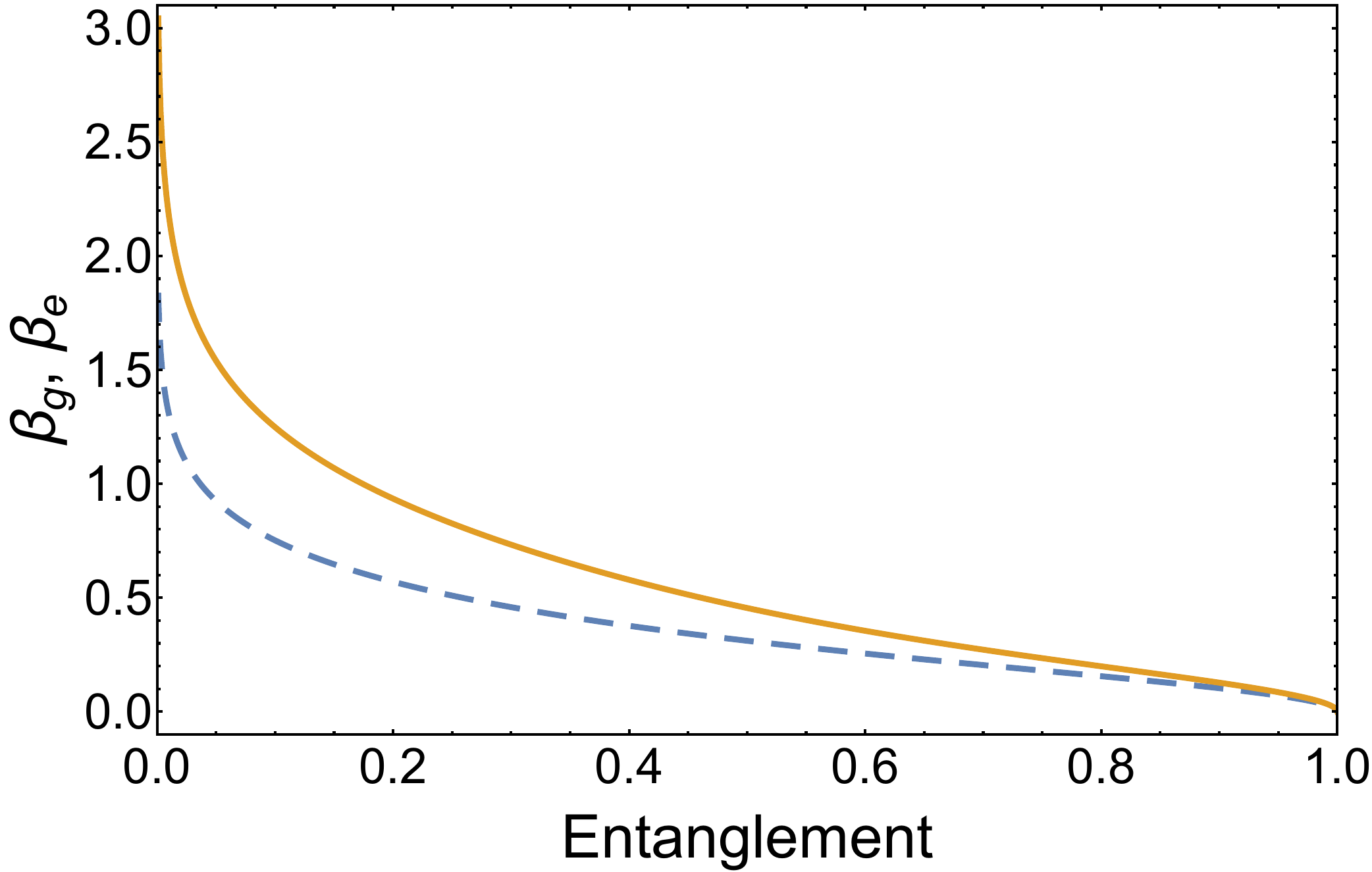}
	\caption{$\betag$  (orange solid line) and $\betae$ (blue dashed line) as functions of the entropy of entanglement (normalized), in the inversed arbitrary units of the local Hamiltonians,  whose spectra are: $\sigma(H_A)=\prtg{0,2,4}$ and $\sigma(H_B)=\prtg{0,1,6,9}$.}
	\label{fig: grafBeta1}
\end{figure}

\section{\label{sec: The convexity problem}   The convexity problem}

As explained in the main text, we need to prove that the lowest energy as a function of the entanglement, $\E$, is monotonically increasing and convex (and analogous properties for the highest energy).
In the following, $\ev{H}= \evn{H}{\psi_g}=\Emin$ and all the other expectation values refer to the state $\stMin$ as a function of $\betag$.
Moreover, we consider $\betag > 0$ as the case $\betag=0$ is obtained only in the extremal case $\E = \ln N_A$.
Then, in the following we consider $\ln \degen < \E < \ln N_A$.
First of all, we calculate
\begin{equation}
\pdv{\betag} \ev{H}= \frac{\prt{\pdvbg \Zmin}^2}{\Zmin^2} - \frac{\pdvbg^2 \Zmin}{\Zmin}.
\end{equation}
One can easily check that
\begin{equation}
\frac{\pdvbg^n \Zmin}{\Zmin} = (-1)^n \ev{H^n}.
\end{equation}
Then,
\begin{equation}
\pdv{\betag} \ev{H}= \ev{H}^2 - \ev{H^2} < 0,
\end{equation}
since the states $\stMin$ are never eigenstates of $H$  for the consider interval of $\E$. If we derive Eq.~\eqref{eq: Thermal solution APP} with respect to $\betag$ we thus obtain
\begin{equation} \label{eq:derivEnt}
\pdvbg \E = \betag \pdvbg \ev{H} < 0.
\end{equation}
This means that $\E(\betag)$ is invertible for $\E > \ln \degen$.
Then, we can use the theorem on the derivative of the inverse function to show that $\ev{H}$ is a monotone function of $\E$:
\begin{equation}\label{eq:1/beta}
\pdve \ev{H} = \pdvbg \ev{H} \times \pdve \betag
= \frac{\pdvbg \ev{H}}{\pdvbg \E}=\frac{1}{\betag}  > 0.
\end{equation}

We can now calculate the second derivative of $\ev{H}$ respect to $\E$:
\begin{equation}
\pdve^2 \ev{H}
=\pdve \prt{\frac{1}{\betag}}
=-\frac{1}{\betag^2} \times \pdve \betag > 0.
\end{equation}

Concerning the maximum energy decreasing monotonicity and concavity with respect to the entanglement, they follow from the fact that the maximum energy is found as the minimum for the Hamiltonian $-(H_A + H_B)$.

We finally observe that all our analysis  extends to other entanglement quantifiers with the following properties:
\begin{enumerate}
	\item they are in a one-to-one relation with the entropy of entanglement when the domain of application is restricted to pure states;
	\item they have the convexity property (see main text);
	\item the minimum (maximum) energy curve is monotonically increasing (decreasing) and convex (concave).
\end{enumerate}
In particular, the first property is enough to prove the main result of Eq.~\eqref{eq: minimal energy entangled state APP}. Concerning the extension of the analysis to the case of mixed states, all these requirements together are a weaker version of the standard assumption made in the main text \cite{Plenio2007}.

This claim becomes intuitive if one thinks that all these requirements let one make the same reasoning done in the main text, with the substitution of the quantifier $\E$ with $\E_m$ when applied to pure states, i.e., calculating the energy curves with respect to $\E (\E_m)$.
Indeed, considering as new quantifier $\E_m$ a bijective function of the quantifier $\E$ is equivalent to apply a transformation along the $x$-axis to the energy-entanglement graphs (such as those of the main text).
Therefore, the results on the mixed states obtained in the main text still apply if this transformation conserves the convexity and concavity of the two energy curves (i.e., the point 3 above) and the chain of relations of Eq.~(12) of the main text becomes:
\begin{equations}
	\Tr \prtg{H\rho}  &\ge \sum_i p_i \Emin(\E (\E_i))   \\
	&\ge \Emin\prt{\sum_i p_i \E (\E_i) } \ge \Emin ( \E(\E_m (\rho))).
\end{equations}

\section{\label{sec: Many-body systems with minimal energy states as ground states}Many-body systems with minimal energy states as ground states}

In this section, we prove that some relevant many-body systems have ground states (obtained by making certain assumption and approximations) belonging to the family of minimum energy states for a fixed degree of entanglement.
Following Ref.~\cite{Peschel2011}, we consider a bipartite system with Hamiltonian $H_T = H_A + H_B + H_I$, where $H_I$ is the interaction Hamiltonian between subsystems $A$ and $B$ and we suppose that $H_I \gg H_A + H_B$, so that the local Hamiltonian $H_A + H_B$ can be treated as a perturbation.
If $H_I$ has a non-degenerate ground state $\ketn{\psi_0}$, perturbation theory gives
\begin{equation}
\ket{\psi_0'} = \ket{\psi_0} - \sum_{k=1}^{N_A N_B - 1} \ket{\psi_k} \frac{\mel{\psi_k}{\left(H_A + H_B\right)}{\psi_0}}{V_k-V_0},
\end{equation}
where $\ketn{\psi_0'}$ is the perturbed ground eigenstate of $H_T$ at first order, $\ketn{\psi_k}$ are the eigenstates of $H_I$ and $V_k$ the eigenvalues.

Every full-rank density matrix~\footnote{A full-rank density matrix has every eigenvalue strictly higher than zero.}
can be written as $\rho = \exp(-\eH)$, where $\eH$ is called \enquote{entanglement Hamiltonian}~\cite{Peschel2011,Laflorencie2016}.
In Ref.~\cite{Peschel2011}, the following result has been proved.
If:
\begin{enumerate}
	\item   there is only one  positive energy $\Delta$  such that  $\mel{\psi_k}{\left(H_A + H_B\right)}{\psi_0} = 0$
	except when $V_k - V_0 = \Delta$,
	\item $\mel{\psi_k}{H_A}{\psi_0} = \mel{\psi_k}{H_B}{\psi_0}$ for each $k$,
	\item $\Tr_{B} \{\dyad{\psi_0}\} \propto I$ (this requires $N_A \leq N_B $ and implies that $\ket{\psi_0}$ is maximally entangled),
\end{enumerate}
then
\begin{equation}\label{eq:HamiltonianEntanglementGroundState}
\rho_A' =\Tr_{B} \{\dyad{\psi_0}\} = \frac{e^{-\eH_A}}{Z} =\frac{1}{Z} \exp(-\frac{4}{\Delta}H_A),
\end{equation}
where $Z=\Tr_A{e^{-\eH_A}}$ assures correct normalization.

The systems analyzed in Ref.~\cite{Peschel2011}, indeed, satisfy the above conditions, but also have the same dimension and $H_A = H_B$.
Therefore, a minimum energy state on these systems takes the form:
\begin{equation}
\stMin=\frac{1}{\sqrt{\Zmin}}\sum_{i=0}^{N_A-1} e^{-\betag A_i} \ket{A_i B_i},
\end{equation}
since $E_i = A_i + B_i = 2 A_i$.
The reduced state of system $A$ is then given by 
\begin{equation}
\rho_A  = \frac{1}{\Zmin} \exp(-2 \betag H_A).
\end{equation}
By comparing the above equation with Eq.~\eqref{eq:HamiltonianEntanglementGroundState}, one sees that the ground states  of these systems are minimum energy states with $\betag = 2/\Delta$.
This can be seen also using in Eq.~(14) of the main text the fact that $\Tr_{B(A)} \{\tilde{H}\} = 2 H_{A(B)}$.

More in general, whenever the first $N_A$ eigenvalues of Hamiltonian $H_B$ are simply connected with the eigenvalues of $H_A$, it is easy to write down the entanglement Hamiltonian form of subsystem $A$ only in terms of $H_A$ when the system is in a minimum energy state.
For example, when $H_B = \alpha H_A$, the entanglement Hamiltonian is given by $\eH_A = \betag (1+\alpha) H_A$ [see Eq.~(14) of the main text].

\section{\label{sec: Minimum energy states are LOCC-connected}   Minimum (maximum) energy states are LOCC-connected}
In this section, we show that minimum energy states are connected by LOCC operations and that the same holds for maximum energy states.
Nielsen's theorem states that a pure state $\ketn{\psi}$ can be converted to a pure state $\ketn{\phi}$ through LOCC if and only if $\lambf^{\psi} \prec \lambf^{\phi}$ (one says that $\lambf^{\phi}$ majorizes $\lambf^{\psi}$), where $\lambf^\psi$ is the vector of the eigenvalues of one of the reduced state of a subsystem for a given state of the total bipartite system~\cite{BookNielsen2010}.
In other words, putting the eigenvalues in non-increasing order, for every $N \leq N_A$
\begin{equation}
\sum_{j=0}^{N-1} \lambda^{\psi}_i
\leq
\sum_{j=0}^{N-1} \lambda^{\phi}_i.
\end{equation}
We now show that the above majorization condition is satisfied between two minimum energy states [see Eq.~\eqref{eq: minimal energy entangled state APP}] when the starting one has not less entanglement than the final one.
If the two minimum states have the same entanglement, then they have the same parameter $\betag$. In this case the majorization condition is trivially satisfied.
Then, let us analyze the case when the two states have a different $\betag$.

We define as follows the sum of the first $N$ elements of the vector $\lambf^{\psi_g} (\betag)$:
\begin{equation}
\lambda_N (\betag) = \frac{\sum_{i=0}^{N-1} e^{-\betag E_i}}{\Zmin},
\qq{where} E_i= A_i + B_i.
\end{equation}
Then, we can show that $\pdvb \lambda_N (\betag) \geq 0\ \forall N$:
\begin{widetext}
	\begin{equations}
		\Zmin^2& \pdvbg \lambda_N (\betag)
		=
		\prt{\sum_{i=0}^{N-1} e^{-\betag E_i}}\prt{\sum_{j=0}^{N_A-1} E_j e^{-\betag E_j}}
		-\prt{\sum_{j=0}^{N_A-1} e^{-\betag E_j}} \prt{\sum_{i=0}^{N-1} E_i e^{-\betag E_i}}\\
		&=
		\prt{\sum_{i=0}^{N-1} e^{-\betag E_i}}\prt{\sum_{j=0}^{N-1} E_j e^{-\betag E_j}+\sum_{j=N}^{N_A-1} E_j e^{-\betag E_j}}
		-\prt{\sum_{j=0}^{N-1} e^{-\betag E_j}+\sum_{j=N}^{N_A-1} e^{-\betag E_j}} \prt{\sum_{i=0}^{N-1} E_i e^{-\betag E_i}}\\
		&=
		\prt{\sum_{i=0}^{N-1} e^{-\betag E_i}}\prt{\sum_{j=N}^{N_A-1} E_j e^{-\betag E_j}}
		-\prt{\sum_{j=N}^{N_A-1} e^{-\betag E_j}} \prt{\sum_{i=0}^{N-1} E_i e^{-\betag E_i}}=
		\sum_{i=0}^{N-1} \sum_{j=N}^{N_A-1} \prt{E_j-E_i}e^{-\betag\prt{E_i+E_j}} \geq 0,
	\end{equations}
\end{widetext}
where the last equality holds because $E_j \geq E_i$ by definition.
This  means that $\betag > \betag' \implies \lambda_N (\betag) \geq \lambda_N (\betag')\ \forall N$, thus proving the majorization condition (the equality sign holds only when $E_0=E_1=\dots=E_{N_A-1}$).
We recall that $\betag$ increases as the entanglement decreases.
Analogously, one can show that maximum energy states can be obtained from one another through LOCC if the starting state has not less entanglement than the final one.

This result has an immediate consequence: for any arbitrary state of the bipartite system one can easily individuate a family of LOCC-connected states including the state at hand.
Every state can be written as follows
\begin{equation}\label{eq:StartingStateLOCC}
\ket{\psi} = \sum_{i=0}^{r-1} \sqrt{\lambda_i} \ket{a_i b_i}, 
\quad \lambda_0 \geq \lambda_1 \geq \dots \lambda_{r-1} > 0,
\end{equation}
where we used the Schmidt decomposition of $\ketn{\psi}$ having Schmidt rank $r$.
Since Nielsen's theorem does not depend on the space containing the two states we want to connect by LOCC  \cite{BookNielsen2010}, we can consider two local Hermitian operators $O_A$ and $O_B$ of dimension $r$ such that:
\begin{equation}\label{eq:MinStatePOV}
\ket{\psi} = \frac{1}{\sqrt{Z_O}}\sum_{i=0}^{r-1} e^{-\frac{\betag}{2} \prt{a_i + b_i}} \ket{a_i b_i},
\end{equation}
where $a_i(b_i)$ are the eigenvalues of the local operator $O_{A(B)}$ with eigenstates $\ketn{a(b)_i}$, and $Z_O$ is a normalization factor.
These Hermitian operators have eigenstates determined by the Schmidt decomposition of $\ketn{\psi}$, but their eigenvalues are only constrained by the condition
\begin{equation}
a_i + b_i =-\frac{1}{\betag}\prt{\ln \lambda_i + \ln Z_O}.
\end{equation}
By also requiring that $O_A$ and $O_B$  have non-decreasing eigenvalues, the state $\ket{\psi}$ can be regarded as a minimum state, in an Hilbert space of dimension $r^2$, with respect to the average of an entire family of operators $O = O_A + O_B$.
It is easy to see that the same state can be considered as a maximum state with respect to the average of operators $O'=-O$. 

Since we have shown that minimum and maximum energy states are, respectively, LOCC-connected among themselves, this also holds  for the states minimizing and maximizing the average of $O$ or $O'$. It follows that, for any given pure state, one can immediately write down a family of states that are connected to it.
For example, starting from the state of Eq.~\eqref{eq:StartingStateLOCC}, corresponding to a given $\betag$, one can obtain through LOCC the following states
\begin{equation}
\label{eq:LOCCfamily}
\ket{\psi' (\betagp)} =   \frac{1}{\sqrt{Z_O}}\sum_{i=0}^{r-1} e^{-\frac{\betagp}{2} \prt{a_i + b_i}} \ket{a_i b_i},
\end{equation}
where $\betagp > \betag$.
On the other hand, the state of Eq.~\eqref{eq:StartingStateLOCC} can be obtained by the above states with $\betagp < \betag$.

The family of states LOCC-connected to an arbitrary state is unique (up to local unitary operators).
Indeed, changing the value of $\betag$ or adding additive constants to the Hermitian operators $O_A$ and $O_B$ does not change the family of reachable states [see Eqs.~\eqref{eq:MinStatePOV} and \eqref{eq:LOCCfamily}].
Moreover, seeing the state $\ketn{\psi}$ as a maximum state determines the exact same family of LOCC-connected states as seeing it as a minimum state since every minimum state of Eq.~\eqref{eq:LOCCfamily} can be seen as a maximum state since $O'=-O$.

\section{\label{sec:Theorem1Application}
	Increasing the energy efficiency of quantum protocols exploiting partial entanglement}
In this section, we show how Theorem~\ref{theorem: minimal energy state for a given set of Schmidt coefficients} can be used to increase the energy efficiency of certain protocols.

In most quantum protocols exploiting partially entangled pure states, the quality of the protocols only depends on the Schmidt coefficients of the entangled states used, where the quality of the protocol is quantified by  quantities such as the fidelity of the result obtained with respect to the desired one or the success probability of the protocol.
For a fixed set of Schmidt coefficients, Theorem~\ref{theorem: minimal energy state for a given set of Schmidt coefficients} provides states having the lowest local energy.
Therefore, for a given protocol with some possible contraints, once the Schmidt coefficients maximizing the quality of the protocol are given, Theorem~\ref{theorem: minimal energy state for a given set of Schmidt coefficients} naturally applies allowing one to allocate the least possible amount of energy on the two subsystems.

For example, in Ref.~\cite{Banaszek2000}, Alice wants to teleport a qudit of dimension $d$ to Bob using a partially entangled state which she shares with Bob.
Their shared state is given by
\begin{equation}
\ket{\psi} = \sum_{i=0}^{d-1} \sqrt{\lambda_i} \ket{a_i b_i},
\end{equation}
where $0\le \sqrt{\lambda_{N_A-1}}\le \dots \le \sqrt{\lambda_1}\le \sqrt{\lambda_0}\le 1$.
The optimal mean fidelity of the quantum teleportation is given by \cite{Banaszek2000}
\begin{equation}
\bar{f} = \frac{1}{d+1}\prtq{1 + \prt{\sum_{i=0}^{d-1} \sqrt{\lambda_i}}^2}.
\end{equation}
Among the states at fixed entropy of entanglement, maximizing the fidelity selects some sets of Schmidt coefficients. Theorem~\ref{theorem: minimal energy state for a given set of Schmidt coefficients} provides the states with the lowest energy for each of these sets. In particular, if we call  the Schmidt coefficients belonging to an optimal set, $\sqrt{\gamma_i}$, where \mbox{$i<j \implies \gamma_i \geq \gamma_j$}.
Then, the state having the lowest local energy with these Schmidt coefficients is
\begin{equation}
\ket{\psi_\textup{Opt}} = \sum_{i=0}^{d-1} \sqrt{\gamma_i} \ket{A_i B_i}.
\end{equation}

More in general, Theorem~\ref{theorem: minimal energy state for a given set of Schmidt coefficients} can greatly simplify maximization problems involving energy bounds since, for the states identified by the theorem, it allows to associate to every squared Schmidt coefficient $\lambda_i$ the energy $E_i=A_i + B_i$. Indeed, using the same example as before, if, for instance, the energy of the shared state provided to Alice and Bob is bounded from above $E_c$, the optimization problem reads
\begin{equation}
\begin{cases}
\ev{\prt{H_A+H_B}}{\psi} \leq E_c,\\
\max_{\ket{\psi}} \bar{f} (\ket{\psi}),
\end{cases}
\end{equation}
while, using Theorem~\ref{theorem: minimal energy state for a given set of Schmidt coefficients} it can be cast in the simplified form
\begin{equation} \label{eq:energyoptimizationproblem}
\begin{cases}
\sum_{i=0}^{d-1} \lambda_i E_i \leq E_c,\\
\max_{\lambf} \bar{f} (\lambf).
\end{cases}
\end{equation}
The search in Eq.~\eqref{eq:energyoptimizationproblem} is much simpler since it is limited to the minimum energy states (for fixed $\lambf$)  selected by Theorem~\ref{theorem: minimal energy state for a given set of Schmidt coefficients}.

\section{\label{sec:EntanglementProduction}    Producing more entanglement with less energy}

Here, we show how producing partially entangled minimum energy states can increase the production of the resource entanglement with respect to the energy spent for generating it, in the LOCC asymptotic limit.

Suppose that we have at our disposal a great number of copies of systems $A$ and $B$ in their ground states with Hamiltonians, respectively, $H_A$ and $H_B$.
We want to increase the production of entanglement between systems $A$ and $B$ with respect to the energy we have to provide to them.
For example, let us analyse the same case considered in Fig.~\ref{fig: grafBeta1} of Appendix~\ref{sec: Lowest energy state for a given entanglement} and Fig.~\ref{fig: comparazione1} of the main text: systems $A$ and $B$ have Hamiltonians with spectra given by $\sigma(H_A)=\prtg{0,2,4}$ and $\sigma(H_B)=\prtg{0,1,6,9}$ in arbitrary units and, initially, all the systems are in their ground state.
To generate a maximally entangled state, we need to give to the bipartite system at least $13/3$ of energy in arbitrary units.
On the other hand, a minimum energy state with energy equal to one half of the latter ($13/6$) has an entanglement equal to roughly $\simeq 0.861$ times the maximal one.
Therefore, creating two states of minimal energy with energy equal to $13/6$ we provide the systems with the same amount of energy but generate about $72\%$ more of the resource entanglement.

We can turn this into a maximization problem.
When both $H_A$ and $H_B$ have degeneracies in their lowest eigenvalue, the problem is trivial since there are entangled states with ground state energy.
Generating $n$ entangled bipartite minimum energy states with entanglement $\E$ costs $E_T= n \Emin (\E)$ of energy (we set $A_0=B_0=0$ for simplicity) and the total amount of entanglement generated is
\begin{equation}
\E_T = n \E = E_T \frac{\E}{\Emin (\E)}.
\end{equation}
Therefore, we want to maximize the ratio $\E/\Emin$.
Using Eq.~\eqref{eq:1/beta} we find
\begin{equation}
\pdv{\E} \prt{\frac{\E}{\Emin}} =
\frac{1}{\Emin}\prtq{1-\frac{\E}{\betag \Emin}} <0,\ \forall \E > 0,
\end{equation}
since $\E = \betag \Emin+\ln[\Zmin]$ and $\Zmin (\E)> 1$ [see Eq.~\eqref{eq: Thermal solution APP} and use $\Emin = - \pdvbg \ln{\Zmin}$].
Therefore, one can generate  more entanglement with the same amount of disposable energy by producing many copies of minimum energy states with lower energy.

\section{Two-mode squeezed states  of two harmonic oscillators}\label{sec: The minimum energy state of two harmonic oscillators is a two-mode squeezed state}

In this section, we show that two-mode squeezed states of two harmonic oscillators belong to the family of our minimum energy states for a given entanglement [see Eq.~\eqref{eq: minimal energy entangled state APP}].
In the main text, we conjecture that our main result on the minimum energy states holds good even for discrete infinite systems.
Then, if this conjecture is correct,  two-mode squeezed states are minimum energy states for a given entanglement.

Consider two harmonic oscillators with Hamiltonians $H_A = \hbar\omega_A \ad a$ and $H_B = \hbar\omega_B \bd b$, where $\omega_{A(B)}$ is their frequency and $\ad (\bd)$ and $a (b)$ are the usual creation and annihilation operators.
In this case, the state of Eq.~\eqref{eq: minimal energy entangled state APP} takes the form (up to phase factors):
\begin{equation}
\stMin=\sqrt{1-e^{-\betag \hbar\omega}}\sum_{n=0}^{\infty} e^{-\frac{\betag \hbar\omega}{2}n} \ket{n_A n_B},
\end{equation}
where $\omega = \omega_A+\omega_B$ and $\ketn{n_{A(B)}}$ are number states in the Fock basis.

A two-mode squeezed state is obtained by applying the following unitary operator on a vacuum state~\cite{Schumaker1985,Adesso2014}:
\begin{equation}\label{eq:squeezing}
U_\textup{sq} = \exp[r\prt{e^{-i\phi}a b -e^{i\phi} \ad \bd}],
\quad r >0,\ \phi \in [0,2\pi[,
\end{equation}
which, by naming $\ketn{\psi_\textup{sq}} = U_\textup{sq} \ketn{0_A 0_B}$, leads to
\begin{equation}
\ket{\psi_\textup{sq}} =
\frac{1}{\cosh(r)} \sum_{n=0}^{\infty} \prtq{-e^{i\phi}\tanh(r)}^n \ket{n_A n_B}.
\end{equation}
We recall that this kind of states are also Gaussian state and are often used in quantum optics laboratories for various tasks, usually exploiting their entanglement~\cite{Adesso2014}.
They can be generated, for instance, through four-wave mixing optical parametric oscillator~\cite{Dutt2015}.

One can easily check that (up to phase factors)
\begin{equation}
\stMin=\ket{\psi_\textup{sq}} \iff \betag \hbar\omega = - \ln\prtq{\tanh[2](r)}.
\end{equation}
Then, if our conjecture holds good, every two-mode squeezed state is also a minimum energy state for a couple of harmonic oscillators.

\section{Dissipative generation of minimum energy states} \label{sec: Generation of a minimum energy state through zero temperature thermalization}

{\it Two-Qubits.---}Suppose to have two qubits with Hamiltonian
\begin{equation}
H_{S} = \frac{\hbar \omega_{A}}{2}\sigma_z^{A}+\frac{ \hbar\omega_{B}}{2}\sigma_z^{B},
\end{equation}
where $\omega_{A(B)}$ is their frequency, $\sigma_z^{A(B)} \ketn{1_{A(B)}} = \ketn{1_{A(B)}}$, and $\sigma_z^{A(B)} \ketn{0_{A(B)}} = -\ketn{0_{A(B)}}$.
The minimum energy state of this system can be easily generated through dissipation assisted techniques.
Assume that one can switch on a suitable interaction so that the new Hamiltonian $H=H_{S}+H_I $ is given by:
\begin{equation}
H=\frac{ \hbar \omega_{A}}{2}\sigma_z^{A}+\frac{ \hbar \omega_{B}}{2}\sigma_z^{B}+\frac{ \hbar g}{2}\prt{\sigma_x^A \sigma_x^B - \sigma_y^A \sigma_y^B}.
\end{equation}
Let us define
\begin{equation}
\omega = \frac{\omega_A + \omega_B}{2} \qq{and}
\delta = \frac{\omega_A - \omega_B}{2}.
\end{equation}
Then, the matrix of the Hamiltonian, in the basis $\prtgn{\ket{1_A 1_B},\ket{1_A 0_B},\ket{0_A 1_B},\ket{0_A 0_B}}$, is:
\begin{equation}
H =\hbar \mqty(
\omega & 0 & 0 & g \\
0 & \delta & 0 & 0 \\
0 & 0 & -\delta & 0 \\
g & 0 & 0 & -\omega \\
).
\end{equation}

The eigenstates and the corresponding eigenvalues are:
\begin{equations}
	&\stMax = \sin \gamma \ket{0_A 0_B}+\cos \gamma \ket{1_A 1_B} , \quad \lambda_e = \hbar \sqrt{\omega^2 + g^2},\\
	&\ket{\psi_+} = \ket{1_A 0_B}, \quad \lambda_+ = \hbar \delta,\\
	&\ket{\psi_-} = \ket{0_A 1_B}, \quad \lambda_- = -   \hbar \delta,\\
	&\stMin =\cos \gamma \ket{0_A 0_B} -\sin \gamma \ket{1_A 1_B} , \: \lambda_g = - \hbar \sqrt{\omega^2 + g^2},
\end{equations}
where $\tan \gamma = \frac{g}{\omega + \sqrt{\omega^2 + g^2}}$. The state $\stMin$ has always the lowest energy. By properly choosing $g$, one can obtain any linear combination (up to phase factors) of $\ketn{1_A 1_B}$ and $\ketn{0_A 0_B}$ desired, i.e., one can obtain a minimum energy state at any given entanglement.

{\it Two-mode squeezed states.---}Here, we reconsider the two-harmonic-oscillator system of Sec.~\ref{sec: The minimum energy state of two harmonic oscillators is a two-mode squeezed state}, showing that two-mode squeezed states can be generated through dissipation.
A two-mode squeezed state can be obtained via application of the unitary squeezing operator, defined in Eq.~\eqref{eq:squeezing}, to the vacuum state of the two oscillators: $U_\textup{sq}(r,\phi)\ketn{0,0}$. It is easy to see that such a state is the ground state of the Hamiltonian $U_\textup{sq}(r,\phi) (H_A + H_B) U_\textup{sq}^\dag (r,\phi)$ which is equal, up to a constant, to
\begin{equation}
H_{(r,\phi)} = \hbar\prtq{  \tilde{\omega}_A \ad a + \tilde{\omega}_B \bd b + g \prt{e^{-i\phi} a b + e^{i \phi} \ad \bd } },
\end{equation}
with $\tilde{\omega}_A = \omega_A \cosh[2](r) + \omega_B \sinh[2](r)$, $\tilde{\omega}_B = \omega_B \cosh[2](r) + \omega_A \sinh[2](r)$, and $g = (\omega_A + \omega_B)\sinh(2r)/2$.
Therefore, the original squeezed state can be obtained as the result of a zero-temperature thermalization when the system is described by the Hamiltonian $H_{(r,\phi)}$. To implement this process, one needs to properly tune the oscillator frequencies and generate the required  interaction term. 


%

\end{document}